\documentclass[onecolumn]{IEEEtran}
\usepackage{amssymb,amsmath,amsfonts}
\usepackage[utf8]{inputenc}
\usepackage[english]{babel}
\usepackage{longtable}
\usepackage{enumerate}

\usepackage{amsthm}

\newtheorem{theorem}{Theorem}
\newtheorem{lemma}{Lemma}
\newtheorem{corollary}{Corollary}
\newtheorem{example}{Example}
\newtheorem{proposition}{Proposition}
\newtheorem{definition}{Definition}
\newtheorem{remark}{Remark}
\usepackage{mathptmx}       
\usepackage{helvet}         
\usepackage{courier}        
\usepackage{type1cm}        
\usepackage{makeidx}         

\usepackage{graphicx}        
\usepackage{multicol}        
\usepackage[bottom]{footmisc}

\usepackage{relsize}

\usepackage{xcolor}
\usepackage{mdframed}

\usepackage{array,etoolbox}
\preto\tabular{\setcounter{magicrownumbers}{0}}
\newcounter{magicrownumbers}

\usepackage[lined,boxed,commentsnumbered,linesnumbered, ruled]{algorithm2e}

\begin{document}

\title{Rate $(n-1)/n$ Systematic\\ 
MDS Convolutional Codes over $GF(2^m)$}

%


\author{\IEEEauthorblockN{\'{A}ngela Barbero}
\IEEEauthorblockA{
\\Universidad de Valladolid\\47011 Valladolid, Spain\\
Email: angbar@wmatem.eis.uva.es\\}
\and
\IEEEauthorblockN{{\O}yvind Ytrehus}
\IEEEauthorblockA{\\Simula@UiB and University of Bergen\\N-5020
Bergen, Norway\\
Email: oyvindy@simula.no}}

\maketitle

\begin{abstract}
A systematic convolutional encoder of rate $(n-1)/n$ and maximum degree $D$ generates a code of free distance at most ${\cal D} = D+2$ and, at best, a column distance profile (CDP) of $[2,3,\ldots,{\cal D}]$. A code is \emph{Maximum Distance Separable} (MDS) if it possesses this CDP. Applied on a communication channel over which packets are transmitted sequentially and which loses (erases) packets randomly, such a code allows the recovery from any pattern of $j$ erasures in the first $j$ $n$-packet blocks for $j<{\cal D}$, with a delay of at most $j$ blocks counting from the first erasure. This paper addresses the problem of finding the largest ${\cal D}$ for which a systematic rate $(n-1)/n$ code over $GF(2^m)$ exists, for given $n$ and $m$. In particular, constructions for rates $(2^m-1)/2^m$ and $(2^{m-1}-1)/2^{m-1}$ are presented which provide optimum values of ${\cal D}$ equal to 3 and 4, respectively. A search algorithm is also developed, which produces new codes for ${\cal D}$ for field sizes $2^m \leq 2^{14}$. Using a complete search version of the algorithm, the maximum value of ${\cal D}$, and codes that achieve it, are determined for all code rates $\geq 1/2$ and every field size $GF(2^m)$ for $m\leq 5$ (and for some rates for $m=6$).
\end{abstract}




\section{Introduction}

\footnote{This work is supported by Ministerio de Economía, Industria y Competitividad, Gobierno de España, through project MTM2013-46949-P, the Estonian Research council through project EMP133, the Norwegian Research Council through the SARDS project. } In many practical communication applications, such as multimedia transmission over  packet erasure channels, on-time delivery is an important quality-of-service criterion. Traditional ARQ systems, for example the one used by TCP for transport layer unicast service, suffer from long delays due to erasures when the round-trip time is large. This has led to an increased interest 
in the design and analysis of systems based on packet-level error correcting codes. Such coded schemes are also known to be beneficial in other transport layer models, for example in the multi-path case.

Two main approaches to this coding problem have been discussed in the literature.
The deterministic approach \cite{Gabidulin1988,StronglyMDS2006,Superregular2013} is to send packets using a fixed $2^m$-ary convolutional code with a good column distance profile. This approach is discussed in Subsection~\ref{subsec:Deterministic}.
Random coding was proposed as a solution in \cite{Tetrys2011,CTCP-2011}. In these schemes, the sender transmits $k$ uncoded information packets, followed by $n-k$ parity check packets formed by random linear combinations of all information packets that have not been acknowledged by the receiver so far. Subsection~\ref{subsec:random} describes this approach, and also discusses a  hybrid approach that combines deterministic and random coding.

\subsection{Contributions}
We present new codes in Section~\ref{sec:newcodes}. In Section~\ref{sec:dist34} we present two new, general, and optimum constructions of MDS convolutional codes. In the literature, there exist only a few general constructions of high-rate convolutional codes: As far as we know, only the Wyner-Ash code \cite{WynerAsh1963} and their binary generalizations (\cite{Ytrehus95}, and Thms. 7.10 and 7.13 in \cite{McElieceHandbook}). We present a simple (but as far as we can see, not previously described in the literature) distance-3 construction. This code has the same rate and Viterbi complexity as the binary Wyner-Ash code, but has a better column distance profile. We also present a much more interesting algebraic distance-4 construction in Proposition~\ref{thm:d4}. In Section~\ref{sec:search} we describe a search algorithm and in Section~\ref{sec:searchresults} we present the codes found by the algorithm. For most parameters, these codes  are better (in a sense which will be made more precise) than previously known codes. Further, we present simple upper bounds in Section~\ref{sec:bounds}.

By convention we will call a convolutional code \emph{systematic} if is it has a systematic encoder; \emph{i.e.} one that preserves all information symbols and obtains redundancy by extra parity symbols.  $2^m$-ary systematic rate $(n-1)/n$ convolutional encoders are useful in order to obtain fast recovery of packet erasures in the common case of channels with moderate erasure rates, and we will focus only on this class of codes.







\section{Background}

\subsection{Notation}

For a thorough introduction to convolutional codes, please see \cite{LinCostello2004}. In the following we will describe the concept of a $2^m$-ary MDS convolutional code in a way which is convenient for our purposes in this paper.

Let $  m \geq 1, n \geq 2, k=n-1$ be integers, $\mathbb{F}=GF(2^m)$, and define the matrices and vectors
\[R_0 = (r_{0,1},\ldots,r_{0,k}) \in \mathbb{F}^k\]
 where $\mathbb{F}^k$ is the $k$-dimensional space of \emph{row} vectors over $\mathbb{F}$,
\[H_0=(R_0|1)\in \mathbb{F}^{n},\]
where $\mathbb{F}^{r \times c}$ denotes the space of matrices with $r$ rows and $c$ columns over $\mathbb{F}$. For $i>1$ define
\[R_i= (r_{i,1},\ldots,r_{i,k}|0) \in \mathbb{F}^k, H_i = \begin{pmatrix}
H_{i-1}\\
R_i\\
\end{pmatrix}\in \mathbb{F}^{(i+1) \times n}\]
and, for an integer $L \geq 2,$ let
\begin{equation}\label{eq:parch}
 H^{(L)} = (H_{L},\begin{pmatrix}
0_{1 \times n}\\
H_{L-1}\\
\end{pmatrix},\ldots,\begin{pmatrix}
0_{(L-1) \times n}\\
H_{0}\\
\end{pmatrix}) \in \mathbb{F}^{(L+1) \times n(L+1)},
\end{equation}
where $0_{r \times c}$ is all-zero matrix with $r$ rows and $c$ columns.
Then $H^{(L)}$ is the parity check matrix for the $L$th truncated block code ${\cal C}^{(L)}$ of a (systematic) convolutional code ${\cal C}$, thus any vector of length $(l+1)n$ for $l \leq L$,
\[[v]_l = (v_1^{(0)},\ldots,v_n^{(0)},v_1^{(1)},\ldots,v_n^{(1)}, \ldots,v_1^{(l)},\ldots,v_n^{(l)}) \in \mathbb{F}^{(l+1)n}\]
is a codeword in ${\cal C}^{(l)}$ if and only if the syndrome
\[H^{(l)} [v]_l^\top  = (0,\ldots,0)^\top \in \mathbb{F}^{(l+1)\times 1}. \]

A systematic encoder for the  code ${\cal C}^{(L)}$ is represented by
\begin{equation}\label{eq:systenc}
G^{(L)} = \begin{pmatrix}
G_0&G_1&\cdots&G_L\\
& G_0 & \cdots & G_{L-1}\\
&     & \ddots & \vdots\\
&     &        & G_0
 \end{pmatrix}\in \mathbb{F}^{k(L+1) \times n(L+1)}
\end{equation}
where
\[G_0 = (I_k | R_0^\top) \in \mathbb{F}^{k \times n}, G_i =(0_k|R_i^\top)\in  \mathbb{F}^{k \times n} \mbox{ for $i>0$},\] and $I_k$ and $0_k$ are the $k \times k$ identity and zero matrices, respectively.  It is straightforward to verify that $G^{(L)}  \times H^{(L)^\top} = 0_{k(L+1) \times (L+1)}$.

\begin{example}
\label{ex-rate23}
Let $\mathbb{F}=GF(2^3)$ with primitive element $\alpha$ defined by $\alpha^3+\alpha+1=0$. Then the parity and generator matrices
\[ H^{(2)} =  \begin{pmatrix}
1 & 1 & 1 & 0 & 0 & 0 & 0 & 0 & 0\\
1 & \alpha & 0 & 1 & 1 & 1 & 0 & 0 & 0\\
\alpha^3 & 1 & 0 & 1 & \alpha & 0 & 1 & 1 & 1  \
 \end{pmatrix} \]
 and
\[ G^{(2)} =  \begin{pmatrix}
1 & 0 & 1 & 0 & 0 & 1 &  0 & 0 & \alpha^3  \\
0 & 1 & 1 & 0 & 0 & \alpha & 0 & 0 & 1  \\
0 & 0 & 0 & 1 & 0 & 1 & 0 & 0 & 1  \\
0 & 0 & 0 & 0 & 1 & 1 & 0 & 0 & \alpha   \\
0 & 0 & 0 & 0 & 0 & 0 & 1 & 0 & 1    \\
0 & 0 & 0 & 0 & 0 & 0 & 0 & 1 & 1

 \end{pmatrix} \]
 define a truncated  code ${\cal C}^{(2)}$, which is rate 6/9 block code over $\mathbb{F}$. Note that the matrices are completely determined by the parity check coefficients $r_{i,j},\; i=0,\ldots,L,\; j=1,\ldots,k$. 
\end{example}

In the conventional polynomial notation of convolutional codes \cite{LinCostello2004}, the parity check matrix can be described as
\[ H(x) = (\sum_{i=0}^{D}r_{i,1}x^i,\ldots,\sum_{i=0}^{D}r_{i,k}x^i,1)\in \mathbb{F}[x].\]
In Example~\ref{ex-rate23}, $H(x)=( 1+x+\alpha^3x^2,1+\alpha x + x^2, 1)$. Similarly, the corresponding polynomial generator matrix is
\[
G(x) =  \begin{pmatrix}
1 & 0 & 1 + x + \alpha^3x^2 \\
0 & 1 & 1 + \alpha x + x^2
\end{pmatrix}
\]

\subsection{MDS convolutional codes constructed from superregular matrices}
\label{subsec:Deterministic}

In the deterministic approach \cite{Gabidulin1988,StronglyMDS2006,Superregular2013}, the goal is to design codes with an optimum \emph{column distance profile,} which we will define below.

The $l$-th column distance $d_l=d_l({\cal C})$ of a convolutional code $\cal C$ is the minimum Hamming weight of any truncated codeword $[c]_l$ \emph{with the first block $(c_1^{(0)},\ldots,c_n^{(0)})$ nonzero}, and the \emph{column distance profile} (CDP) is the non-decreasing sequence $(d_0,d_1,d_2,\ldots,d_D = {\cal D}, {\cal D}, {\cal D}, \ldots)$, where ${\cal D}$ is the free distance of the code and $D$ is the index for which the CDP reaches ${\cal D}$. The CDP was originally studied for its significance on the performance of sequential decoding (please see Ch. 13 of \cite{LinCostello2004}.) Recently the CDP has received renewed attention in the context of $2^m$-ary codes, due to its importance for fast recovery from losses of symbols in an erasure channel.

Recall that we consider only convolutional codes of rate $k/n = (n-1)/n$ that have a systematic encoder. In this case, by the Singleton bound for truncated block codes, $d_0 \leq 2$, and by similar linear algebra arguments, $d_l \leq d_{l-1}+1$ for $l>0$. Moreover, $d_l = d_{l-1}$ for $l>D$. So the best column distance profile one can hope to find in a code with a systematic encoder  is
\begin{equation}
\label{eq_MDS}
 d_0 = 2, d_1 = 3, \ldots, d_j = j+2, \ldots, d_D = D+2 = {\cal D}.
\end{equation}

By an \emph{MDS convolutional code},  in this paper we will mean a code with a CDP as in (\ref{eq_MDS}).



\begin{remark}
The concept of Strongly-MDS codes was introduced in \cite{StronglyMDS2006}. This concept takes into account that for some codes that do not possess a systematic encoder, the free distance may grow beyond $\delta+2$, where $\delta$ is the memory of a \emph{minimal} encoder. In order not to complicate the notation, and since Viterbi complexity is not an issue in this paper, we omit the details.
\end{remark}

\begin{definition}
\label{def:superregular}
Consider a lower triangular matrix

\[SR=\left(\begin{array}{ccccc}
r_0 & 0 & 0 & \cdots & 0 \\
r_1 & r_0 & 0 & \cdots & 0 \\
r_2 & r_1 & r_0 & \cdots & 0 \\
\vdots & \vdots & \vdots & \ddots & \vdots \\
r_L & r_{L-1} & r_{L-2} & \cdots & r_0
\end{array} \right)\]
where each element $r_i\in {\mathbb F}$.

Consider a square submatrix $P$ of size $p$   of $SR$, formed by the entries of $SR$ in the rows with indices $1\leq i_1 < i_2<\cdots <i_p\leq (L+1)$ and columns of indices $1\leq j_1 < \cdots < j_p \leq (L+1)$. $P$, and its corresponding minor, are proper if $j_l\leq i_l$ for all $l\in\{ 1, \ldots, p\}$.

$SR$ is superregular if all its proper $p\times p$ minors are non singular for any $p\leq L+1$.

When matrix $SR$ is upper triangular the definition of proper submatrices is analogous.


%
%
\end{definition}



A $\gamma \times \gamma$ superregular matrix can be used to construct a rate $1/2$ code in two ways: (1) \cite{Gabidulin1988} a systematic MDS convolutional code with  CDP as in (\ref{eq_MDS}) with ${\cal D}=D+2 = \gamma+1$, (2) \cite{StronglyMDS2006} when $\gamma = 2\delta+1$,  a strongly-MDS code (in general nonsystematic) with a parity check matrix of max degree $\delta$ and the same CDP as for the systematic codes in case (1).


While superregular matrices are known to exist for all dimensions if the field is large enough, general efficient constructions are not known, and for $\gamma \gtrsim 10$ the minimum field size for which a $\gamma \times \gamma$ superregular matrix exists is not known.  Another problem with the deterministic approach is that the existing design methods do not allow a simple construction of codes of high rate and/or high degree. Codes of \emph{higher} rates (which are desirable in many practical cases) can also be constructed from these superregular matrices, but this involves deleting columns, so that the conditions on a superregular matrix are too strict. This means that in practice only  simple 
codes can be constructed in this way. Since superregular matrices are so hard to construct, the reduction to the superregular matrix problem blocks the code construction. Therefore we generalize Definition~\ref{def:superregular} as follows:

\begin{definition}
\label{def:ksuperregular}
Consider an $s$\emph{-lower triangular} matrix (where $s$ is a positive integer)

\begin{equation}\label{eq:s-superregular}
SSR = \left(
  \begin{array}{ccccccccccccc}

  r_{0,1} & \cdots & r_{0,s} & 0 & \cdots & 0 &             0 & \cdots & 0 & \cdots  & 0 & \cdots & 0 \\
  r_{1,1} & \cdots & r_{1,s}  & r_{0,1} & \cdots & r_{0,s} & 0 & \cdots & 0 & \cdots & 0 & \cdots & 0 \\
  r_{2,1} & \cdots & r_{2,s}  & r_{1,1} & \cdots & r_{1,s} & r_{0,1} & \cdots & r_{0,s} & \cdots& 0 & \cdots & 0 \\
  \vdots & \ddots & \vdots & \vdots & \ddots & \vdots & \vdots & \ddots & \vdots & \cdots & \vdots & \ddots & \vdots \\
  r_{L-1,1} & \cdots & r_{L-1,s} & r_{L-2,1} & \cdots & r_{L-2,s}& r_{L-3,1}& \cdots & r_{L-3,s} & \cdots & 0 & \cdots & 0 \\
  r_{L,1} & \cdots & r_{L,s} & r_{L-1,1} & \cdots & r_{L-1,s}& r_{L-2,1}& \cdots & r_{L-2,s} & \cdots & r_{0,1} & \cdots & r_{0,s} \\

    \end{array}
\right)
\end{equation}

Consider a square submatrix $P$ of size $p$   of $SSR$, formed by the entries of $SSR$ in the rows with indices $1\leq i_1 < i_2<\cdots <i_p\leq (L+1)$ and columns of indices $1\leq j_1 < \cdots < j_p \leq s(L+1)$. $P$, and its corresponding minor, are proper if
$j_l\leq s\cdot i_l$ for all $l\in\{ 1, \ldots, p\}$.

The matrix $SSR$  is called $s$-\emph{superregular} iff all of its proper $p \times p$ minors, for any $p\leq L+1$, are nonsingular.

\end{definition}

The following lemma is a restatement of Theorem~1 in \cite{Gabidulin1988}, using the terminology of this section.
\begin{lemma}
\label{lemma:biglemma}
Let $H^{(D)}$ be the parity check matrix of the $D-th$ truncation of a systematic convolutional code, given by

\begin{equation}\label{}
H^{(D)} = \left(
  \begin{array}{ccccccccccccccccc}

  r_{0,1} & \cdots & r_{0,k} & 1 & 0 & \cdots & 0 &  0 &            0 & \cdots & 0 & 0 & \cdots  & 0 & \cdots & 0 &0 \\
  r_{1,1} & \cdots & r_{1,k} & 0  & r_{0,1} & \cdots & r_{0,k} & 1 &  0 & \cdots & 0 & 0 & \cdots & 0 & \cdots & 0 &0 \\
  r_{2,1} & \cdots & r_{2,k} & 0 & r_{1,1} & \cdots & r_{1,k} & 0 & r_{0,1} & \cdots & r_{0,k} & 1 & \cdots& 0 & \cdots & 0 & 0 \\
  \vdots & \ddots & \vdots &  \vdots & \vdots & \ddots & \vdots & \vdots & \vdots & \ddots & \vdots & \vdots & \cdots & \vdots & \ddots & \vdots & \vdots \\
  r_{D-1,1} & \cdots & r_{D-1,k} & 0 &  r_{D-2,1} & \cdots & r_{D-2,k}& 0 &  r_{D-3,1}& \cdots & r_{D-3,k} & 0 & \cdots & 0 & \cdots & 0 & 0 \\
  r_{D,1} & \cdots & r_{D,k} & 0&  r_{D-1,1} & \cdots & r_{D-1,k}& 0 & r_{D-2,1}& \cdots & r_{D-2,k} &0  & \cdots & r_{0,1} & \cdots & r_{0,k} & 1 \\

    \end{array}
\right)
\end{equation}

and let $H'^{(D)}$ be the matrix obtained from $H^{(D)}$ by removing the columns in positions $(k+1), 2(k+1), 3(k+1), \ldots , (D+1)(k+1)$, that is

\begin{equation}\label{eq:ksuperregular}
H'^{(D)} = \left(
  \begin{array}{ccccccccccccc}

  r_{0,1} & \cdots & r_{0,k} & 0 & \cdots & 0 &             0 & \cdots & 0 & \cdots  & 0 & \cdots & 0 \\
  r_{1,1} & \cdots & r_{1,k}  & r_{0,1} & \cdots & r_{0,k} & 0 & \cdots & 0 & \cdots & 0 & \cdots & 0 \\
  r_{2,1} & \cdots & r_{2,k}  & r_{1,1} & \cdots & r_{1,k} & r_{0,1} & \cdots & r_{0,k} & \cdots& 0 & \cdots & 0 \\
  \vdots & \ddots & \vdots & \vdots & \ddots & \vdots & \vdots & \ddots & \vdots & \cdots & \vdots & \ddots & \vdots \\
  r_{D-1,1} & \cdots & r_{D-1,k} & r_{D-2,1} & \cdots & r_{D-2,k}& r_{D-3,1}& \cdots & r_{D-3,k} & \cdots & 0 & \cdots & 0 \\
  r_{D,1} & \cdots & r_{D,k} & r_{D-1,1} & \cdots & r_{D-1,k}& r_{D-2,1}& \cdots & r_{D-2,k} & \cdots & r_{0,1} & \cdots & r_{0,k} \\

    \end{array}
\right)
\end{equation}

Then the CDP of the convolutional code given by $H^{(D)}$ is $(2,3,\ldots , D+2)$ if and only if $H'^{(D)}$ is a $k$-superregular matrix.

\end{lemma}

Theorem~1 in \cite{Gabidulin1988} is stated without proof. For reference, we include a formal proof in Appendix~A.


\begin{definition}
\label{def:Delta}
Let $\Delta(2^m, n)$ be the largest free distance ${\cal D}$ such that there exists a rate $(n-1)/n$  systematic MDS convolutional code over $GF(2^m)$ with column distance profile as in (\ref{eq_MDS}).
\end{definition}




The main problem that we address in this paper is to determine exact values, or constructive lower bounds, for $\Delta(2^m, n)$. Please note that there is no restriction of the degree $D$  in Definition~\ref{def:Delta}.

There are few known code constructions in the literature, beyond those based on superregular matrices. Table~\ref{tab:old} contains the current world records with respect to rate $(n-1)/n$ MDS codes, to the best of our knowledge. We will describe new codes in Section~\ref{sec:newcodes}.
\begin{table}
\begin{center}
\begin{tabular}{|r|c|c|c|}
\hline
  Rate & Field size  & $\cal D$ & Description\\
  \hline
1/2 &        4 &    4 &    \cite{JustesenHughes1974} \\
1/2 &        8 &    6 &    \cite{JustesenHughes1974} \\
1/2 &        8 &    6 &    \cite{StronglyMDS2006} superregular\\
1/2     &    32 &    8 &    \cite{StronglyMDS2006} superregular\\
2/3     &    16 &    5 &    \cite{StronglyMDS2006} ad hoc\\
2/3     &    64    & 5 &    \cite{StronglyMDS2006} superregular\\
3/4     &    16 &    3 &    \cite{StronglyMDS2006} superregular\\
\hline
\end{tabular}
\vspace{0.1cm}
\caption{Some rate $(n-1)/n$ MDS codes (not necessarily systematic) described in the literature. }
\label{tab:old}
\end{center}
\end{table}

Although this paper focuses on rate $(n-1)/n$ MDS codes, we observe that the following lemma, that follows directly from Theorem~2 in \cite{Gabidulin1988}, implies that our results will also provide rate $1/n$ MDS codes.

\begin{lemma}
If a systematic rate $(n-1)/n$ MDS code of memory $D$ and free distance $D+2$ exists, then its dual code is equivalent to a systematic rate $1/n$ MDS code of  memory $D$ and free distance $(n-1)(D+1)+1$.
\end{lemma}

\subsection{Random convolutional codes}
\label{subsec:random}

In the terminology of this paper, the random approach \cite{Tetrys2011,CTCP-2011} consists of selecting the coefficients of $r_{ij}$  independently at random. The advantage of this is that one can pick codes with large degrees, and that over large fields the expected performance is ``reasonably good'', although the exact loss compared to optimum average performance or optimum guaranteed worst case performance remains to be determined.

Coefficients need to be transmitted in the headers of the data packets, but this represents only a small rate loss when large packets are transmitted.

\begin{proposition}
Consider a hybrid scheme where the first blocks of coefficients $r_{i,j}$ (until time $i=D$) are selected fixed, and subsequent random coefficients $r_{i,j}$ for $i>D$ are selected at random.  Thus the parity check equation will be on the form
\begin{equation}
H(x)= H_{CDP}(x) + H_{Random}(x)
\end{equation}
where
\[
H_{CDP}(x) = (\sum_{i=0}^{D}r_{i,1}x^i,\ldots,\sum_{i=0}^{D}r_{i,k}x^i,1)
\]
and
\[
 H_{Random}(x)(\sum_{i=D+1}^{?}R_{i,1}x^i,\ldots,\sum_{i=D+1}^{?}R_{i,k}x^i,0)
\]
where all $R_{i,j}$ are nonzero randomly selected coefficients and where the degree of the random polynomials does not need to be fixed (except by the application protocol). Then the initial CDP (until time $D$) is not affected by the random part of the code construction.
\end{proposition}
\begin{proof}
Obvious: Only the first component $H_{CDP}(x)$ of the parity check matrix determines the initial part of the CDP.
\end{proof}

\emph{Our suggestion} is to use such hybrid codes, \emph{i. e.}  codes where the terms of degree $0,\ldots,D$ of the parity check polynomials are preselected constants yielding an optimum initial column distance profile, while subsequent random parity checks are added as needed. This guarantees optimum recovery for the simplest and most likely erasure patterns, and hence  better performance than random codes for light to moderate erasure patterns, while still allowing the degree to grow if required by the application.

\section{New codes}
\label{sec:newcodes}
Gluesing-Luerssen \emph{et. al.} \cite{StronglyMDS2006}  use  superregular matrices to design codes. However, the authors also give examples of codes that are better than the ones constructed from superregular matrices, and note that "..the abundance
of (small) examples suggests that such a construction
might be possible and might lead to smaller alphabets for given
parameters than the construction ,[...]  We will
leave this as an open question for future research."

So here comes the future research. In this section we present constructions and a new search algorithm that, in combination, improve our knowledge of  $\Delta(2^m,n)$ for almost all sets of parameters, with respect to what we find in the literature.






\subsection{Codes with free distance ${\cal D} \in \{3,4\}$}
\label{sec:dist34}

We present two optimum constructions, for ${\cal D} \in \{3,4\}$. For ${\cal D}=3$ the construction is simple, but we have not seen it presented in prior literature.




We have tacitly assumed the following fact for the constant terms. Here comes the justification.
\begin{lemma}
\label{lem:constantone}
We can w.l.o.g assume $r_{0,1}=\cdots=r_{0,n}=1$.
\end{lemma}
\begin{proof}
If there is a  $r_{0,j}$ equal to zero, then $d_0 <2$. We don't want that.
Then assume some nonzero $r_{0,j} \neq 1$. If we multiply the corresponding column of $G^{(D)}$ by $r_{0,j}^{-1}$, we obtain a new code with the same CDP and weight structure.
\end{proof}

\begin{proposition}
\label{thm:d3}
$\Delta(q^m,q^m) = 3$  for $q$   prime  and $m \geq 0$.
\end{proposition}

\begin{proof}
Select $r_{0,i}=1$ and $r_{1,i}$, $i=1,\ldots,q^m-1$ as the $q^m-1$ distinct nonzero elements of $GF(q^m)$. Without loss of generality,  the parity check matrix of (\ref{eq:parch}) takes the form
\[
 H^{(1)} = \begin{pmatrix}
1 & 1 & \cdots & 1   & 1 &  0 & \cdots & 0 & 0\\
1 & 2 & \cdots & q^m-1 & 0 &  1 & \cdots & 1 & 1\\
\end{pmatrix}
\]

\[H'^{(1)}= \begin{pmatrix}
1 & 1 & \cdots & 1   &  0 & \cdots & 0\\
1 & 2 & \cdots & q^m-1 &  1 & \cdots & 1\\
\end{pmatrix}
\]
 is $q^m$-superregular because it is obvious that  all the proper minors of sizes 1 and 2 are nonsingular.
Clearly, $d_0$ = 2 and $d_1$ = 3.
\end{proof}


\begin{remark}
It is instructive to compare the construction of Proposition~\ref{thm:d3} with the binary Wyner-Ash codes \cite{WynerAsh1963}. Wyner-Ash codes were considered for digital media transmission already in 1974 \cite{bbcWynerAsh1974}. The Wyner-Ash code of length 4 has the binary polynomial parity check matrix
\[H_{WA}=\left(\begin{array}{cccc}
1+x+x^2 & 1+x & 1+x^2 & 1
\end{array} \right).\]
It is easy to see that the CDP of the Wyner-Ash code is $[2,2,3]$, \emph{i. e.} this is \emph{not} an MDS code. The construction of Proposition~\ref{thm:d3} can be considered as a $q^m$-ary generalization of the Wyner-Ash code, of memory 2, but this code \emph{is} an MDS code, with CDP $[2,3]$.
\end{remark}

For ${\cal D}=4$, we present an  optimum construction in Proposition~\ref{thm:d4}. Complete computer searches for $m\leq 5$ indicate that the construction is unique and, in a sense, much better than what can be achieved through other choices of the set of first degree coefficients $\{r_{1,i}\}$.
\begin{lemma}
\label{lemma-condition-D4}
For a code with a CDP of $[2,3,4]$,  its parity check matrix $H^{(2)}$ must satisfy \\
(i) $r_{i,s} \neq 0$ for $i=1,2$, $s=1,\ldots,k$, \\
(ii) $r_{i,s} \neq r_{i,t}$ for $i =1,2$, $1 \leq s < t \leq k$,\\
(iii) $r_{1,t}\neq r_{2,s}/r_{1,s}$ for $1\leq s,t \leq k$,\\
(iv)  $r_{2,s}/r_{1,s} \neq r_{2,t}/r_{1,t}$ for $1 \leq s <t \leq k$,\\
(v) $r_{2,s} - r_{2,t} \neq   r_{1,u}(r_{1,s}-r_{1,t})$ for $1 \leq s < t  \leq k$, $1 \leq u  \leq k$,\\
(vi) $r_{2,s}  \neq (r_{1,s}(r_{2,u}-r_{2,t})-r_{1,t}r_{2,u}+r_{1,u}r_{2,t})/(r_{1,u}-r_{1,t})$ for $1 \leq s<t<u \leq k$.
\end{lemma}

\begin{proof}

From Lemma~\ref{lemma:biglemma} we need $H'^{(2)}$ to be $k$-superregular.

That all $1\times 1$ proper minors of $H'^{(2)}$ are non singular is equivalent to condition (i).

Proper minors of size $2\times 2$ are of the following types

\[ \begin{vmatrix}
\label{}
1 & 0\\
r_{i,s} &1
 \end{vmatrix} \mbox { , }
 \begin{vmatrix}
\label{}
1 & 0\\
r_{2,s} & r_{1,t}
 \end{vmatrix} \mbox { , }
 \begin{vmatrix}
\label{}
1 & 1\\
r_{i,s} & r_{i,t}
 \end{vmatrix} \mbox { , }
   \begin{vmatrix}
\label{}
r_{1,s} & 1\\
r_{2,s} & r_{1,t}
 \end{vmatrix} \mbox { , }
 \begin{vmatrix}
\label{}
r_{1,s} & r_{1,t}\\
r_{2,s} & r_{2,t}
 \end{vmatrix} \]

The first type are trivially non zero, the second type are non zero when condition (i) is satisfied. The third type being nonsingular is equivalent to condition (ii), the fourth type is guaranteed to be non zero if and only if condition (iii) is satisfied  and the fifth type being nonsingular is equivalent to condition (iv).

Finally, $3\times 3$ proper minors can be of four different types:

\[\begin{vmatrix}
\label{}
1 & 0 & 0\\
r_{1,s} &1 & 0\\
r_{2,s} & r_{2,t} & 1
 \end{vmatrix} \mbox { , }
\begin{vmatrix}
\label{}
1 & 1 & 0\\
r_{1,s} & r_{1,t} & 0\\
r_{2,s} & r_{2,t} & 1
 \end{vmatrix} \mbox { , }
\begin{vmatrix}
\label{}
1 & 1 & 0\\
r_{1,s} & r_{1,t} & 1\\
r_{2,s} & r_{2,t} & r_{1,u}
 \end{vmatrix} \mbox { , }
\begin{vmatrix}
\label{}
1 & 1 & 1\\
r_{1,s} &r_{1,t} & r_{1,u}\\
r_{2,s} & r_{2,t} & r_{2,u}
 \end{vmatrix}\]

 Those of the first type are trivially nonsingular. Condition (ii) takes care of those in the second type to be nonsingular. Those in the third type are nonsingular if and only if condition (v) is satisfied. The fifth type are non singular if and only if condition (vi) is satisfied.

\end{proof}
%

%

\begin{example}
Consider the code in Example~\ref{ex-rate23}. By checking conditions (i)-(vi) in Lemma~\ref{lemma-condition-D4} we observe that the code has CDP equal to [2,3,4].
\end{example}

\begin{proposition}
\label{thm:d4}
$\Delta(2^m,2^{m-1}) = 4.$
\end{proposition}

\begin{proof}
Let $\mathbb{F}=GF(2^m)$.

($\geq:$) The following construction gives a code that meets the requirements:\
The trace function \cite{MacWSloane77} is defined by
\[\begin{array}{cccl}
Tr^m() : &\mathbb{F} & \rightarrow &GF(2)\\
  & x & \rightarrow &  Tr^m(x) = \sum_{i=0}^{m-1}x^{2^i}. \end{array}\]
Consider the set
\[H_\beta = \{ x \in \mathbb{F} | Tr^m(\beta x) = 0\}.\]
When $\mathbb{F}$ is regarded as an $m$-dimensional vector space over $GF(2)$, the set $H_\beta$ 
is a hyperplane (an $(m-1)$-dimensional linear subspace) of $\mathbb{F}$. Let $k=2^{m-1}-1$, select $\beta$ 
as an arbitrary nonzero field element, select $c$ as an arbitrary constant in $\mathbb{F} \setminus H_\beta$. 
Then select $a_1,\ldots,a_k := r_{1,1},\ldots,r_{1,k}$ as all distinct nonzero elements in $H_{\beta}$, 
and set $b_s := r_{2,s}= a_s(a_s+c) = r_{1,s}(r_{1,s}+c)$ for $s=1,\ldots,k$. We need to verify that this construction satisfies the conditions in Lemma~\ref{lemma-condition-D4}. \\
\emph{(i)} This holds because $b_s = a_s(a_s+c)$ is a product of two nonzeros.\\
(\emph{ii)} All $a_s$'s are distinct. Assume that $b_s = b_t, s \neq t$. Then $0 = a_s(a_s+c) = a_t(a_t+c)= (a_s+a_t)c+a_s^2+a_t^2 = (a_s+a_t)c+(a_s+a_t)^2 = (a_s+a_t)(c+a_s+a_t)$. The first factor is nonzero since $a_s \neq a_t$. The second factor is also nonzero since $a_s+a_t \in H_\beta$ (because $H_\beta$ 
is closed under addition) while $c \not \in H_\beta$, 
a contradiction. \\
\emph{(iii)} Assume that $ a_sa_t = b_s$. Then $a_sa_t=a_s(a_s+c) \Rightarrow a_t = a_s+c,$ a contradiction, since $a_t \in H_{\beta}$ and $a_s+c\not\in H_{\beta}.$\\
\emph{(iv)} Assume that $ b_s/a_s = b_t/a_t, s \neq t$. Then $a_s + c = a_t + c \Rightarrow a_s = a_t,$ a contradiction.\\
\emph{(v)}
\begin{eqnarray} 
\nonumber b_s+b_t+a_u(a_s+a_t) &=&a_s(a_s+c)+a_t(a_t+c)+a_u(a_s+a_t) \\
\nonumber   &=& a_s^2+a_t^2+(a_s+a_t)(c+a_u)\\
\nonumber    &=&(a_s+a_t)^2+(a_s+a_t)(c+a_u)\\
\nonumber   &=& (a_s+a_t)(a_s+a_t+c+a_u)\end{eqnarray} which again is a product of nonzero factors, because $c \not \in H_\beta$ and  $a_s+a_t+a_u \in H_\beta$, and hence nonzero.\\
\emph{(vi)}
\begin{eqnarray}
\nonumber b_s+\frac{a_s(b_t+b_u) +a_ub_t+a_tb_u}{a_t+a_u} &=& a_s(a_s+c)+\frac{a_s(a_t(a_t+c)+a_u(a_u+c))+a_ua_t(a_t+c)+a_ta_u(a_u+c)}{a_t+a_u} \\
\nonumber  &=&a_s(a_s+c)+\frac{a_s(a_t+a_u)^2+a_sc(a_t+a_u)+a_ta_u(a_t+a_u)}{a_t+a_u} \\
\nonumber  &=& a_s(a_s+c)+a_s(a_t+a_u)+a_sc+a_ta_u \\
\nonumber  &=& a_s^2+a_sa_t+a_sa_u+a_ta_u\\
\nonumber  &=& (a_s+a_t)(a_s+a_u) \neq 0.
\end{eqnarray}


($\leq:$) This follows from Theorem~\ref{th-upb-1} in Section~\ref{sec:bounds}.

\end{proof}

\begin{remark}
By Theorem~\ref{th-upb-1} later, the construction in Proposition~\ref{thm:d4} is optimum not only in the sense that it offers the maximum  distance 4 for the given field size and code rate, but also it offers the minimum field size for a code of the given rate and distance 4, and the maximum code rate given the field size and distance 4. Moreover, complete computer search for field sizes $2^m \leq 32$ show that the construction is unique for these parameters.
\end{remark}


\subsection{Computer search algorithm}
\label{sec:search}

The goal of the search algorithm is to  select the coefficients $r_{i,j}$ successively, ordered first on $i$ and then reversely on $j$, in such a way that the conditions on the minors are met. 

\subsubsection{Some useful facts}

First, as for the constructions in Section \ref{sec:dist34}, we use Lemma \ref{lem:constantone} in order to set $r_{0,1}=\cdots=r_{0,k}=1$. In order to simplify the search we apply the following results.



\begin{lemma}
\label{lem:ordered}
We can w.l.o.g assume $r_{1,i}<r_{1,i+1}$, $i=1,\ldots,k-1$ for any choice of ordering $<$.
\end{lemma}

\begin{lemma}
\label{lem:r11}
Consider an MDS convolutional code ${\cal C}$ with polynomial parity check matrix
\[ H(x) = (1+\sum_{i=1}^{D}r_{i,1}x^i,\ldots,1+\sum_{i=1}^{D}r_{i,k}x^i,1)\in \mathbb{F}[x].\]

Then the code ${\cal C}_c$  with parity check matrix

\[ H_c(x) = (1+\sum_{i=1}^{D}c^{i}r_{i,1}x^i,\ldots,1+\sum_{i=1}^{D}c^{i}r_{i,k}x^i,1)\in \mathbb{F}[x]\]

is also  MDS for any $c \in \mathbb{F} \setminus \{0\}$.
\end{lemma}

\begin{proof}
Let $v(x)=(v_1(x),\ldots,v_n(x)) = (\sum_{i=0}^D v_{1,i}x^i,\ldots,\sum_{i=0}^Dv_{n,i}x^i)$. Then $v(x)H(x)^\top=0$ iff $v_c(x)H_c(x)^\top=0$  for
\[v_c(x) = (\sum_{i=0}^D c^{-i}v_{1,i}x^i,\ldots,\sum_{i=0}^Dc^{-i}v_{n,i}x^i).\]
\end{proof}

\begin{corollary}
\label{cor:roneisone}
If a systematic MDS convolutional code exists, we can w.l.o.g. assume that it has a parity check matrix with $r_{1,k}=1$.
\end{corollary}
\begin{proof}
Assume that a systematic MDS convolutional code exists with a parity check matrix with $r_{1,k}=a \in \mathbb{F}$. Then apply Lemma ~\ref{lem:r11} with $c=a^{-1}$.
\end{proof}

\begin{lemma}
\label{lem:cyclotomic1}
Let $M$ be a $k$-superregular matrix over $GF(q^m)$, with $q$ a prime. Raising each element of $M$ to power $q$ yields another $k$-superregular matrix.
\end{lemma}

\begin{proof}

Given any square matrix
\[A=\begin{pmatrix}
a_{1,1} & \cdots & a_{1,n}\\
\vdots &  & \vdots \\
a_{n,1} & \cdots & a_{n,n}\\
\end{pmatrix}
\]
 with $a_{i,j}\in GF(q^m)$, by definition
 \[\det(A)=\sum_{\sigma\in S_n} s(\sigma)a_{1,\sigma(1)}\cdots a_{n,\sigma(n)}\]

 where $S_n$ is the group of permutations of $n$ elements and $s(\sigma)$ is the sign of each permutation $\sigma$.

 Also we have
 \[(x_1+x_2+\cdots +x_n)^q=\sum_{\begin{array}{c} q_1+\cdots +q_n=q \\ 0\leq q_i\leq q \end{array}} c(q_1,\ldots, q_n)x_1^{q_1}x_2^{q_2}\cdots x_n^{q_n}\]
 and $c(q_1, \ldots , q_n)=\binom{q}{q_1}\binom{q-q_1}{q_2} \cdots \binom{q-q_1-\cdots -q_{n-1}}{q_n}$.

 When $q$ is prime, $q$ is a divisor of  coefficient $c(q_1, \ldots, q_n)$  except in the cases $q_i=q, q_j=0$ for $j\neq i$, and this for each $i\in \{1, \ldots , n\}$. Therefore, in characteristic $q$ we have

 \[(x_1+x_2+\cdots +x_n)^q=x_1^q+ x_2^q + \cdots + x_n^q\]

 Back to the definition of determinant, in characteristic $q$ we have

 \[\left(\det(A)\right)^q=\left(\sum_{\sigma\in S_n} s(\sigma) a_{1,\sigma(1)}\cdots a_{n,\sigma(n)}\right)^q=\sum_{\sigma\in S_n} s(\sigma)^q a_{1,\sigma(1)}^q\cdots a_{n,\sigma(n)}^q\]

 Finally, $s(\sigma)$ is either 1 or -1, so in case $q=2$ then $s(\sigma)=1$ for any $\sigma \in S_n$, and $s(\sigma)^2=s(\sigma)$,  and in case $q$ is odd we have $s(\sigma)=s(\sigma)^q$.

 This gives

 \[\left(\det(A)\right)^q=\sum_{\sigma\in S_n} s(\sigma)^q a_{1,\sigma(1)}^q\cdots a_{n,\sigma(n)}^q=\sum_{\sigma\in S_n} s(\sigma) a_{1,\sigma(1)}^q\cdots a_{n,\sigma(n)}^q=\det(A^{\circ q})\]

 where $A^{\circ q}$ denotes the $q$-th Hadamard or Schur power of $A$, that is, the matrix whose entries are the the entries of $A$ raised to $q$.

 Now it is clear that given any proper minor $P$ of size $p$ of matrix $M$ in $GF(q^m)$, $P^{\circ q}$ is the corresponding proper minor in $M^{\circ q}$ and $P$ is nonsingular if and only if $P^{\circ q}$ is nonsingular, so $M$ is $k$-superregular if and only if $M^{\circ q}$ is k-superregular.

\end{proof}

\begin{corollary}
\label{lem:cyclotomic2}
In particular, let $M$ be a $k$-superregular matrix over $GF(2^m)$. Squaring each element of $M$ yields another $k$-superregular matrix.
\end{corollary}

\begin{proof}

This is just the particular case for $q=2$ of the Lemma above.

\end{proof}

Hence we also have:

\begin{corollary}
\label{cor:cyclotomic}
Assume that the values for $r_{0,i},\; i=1,\ldots,k$ and for $r_{1,k}$ are all fixed to $1$, as allowed by Lemma~\ref{lem:constantone} and Corollary~\ref{cor:roneisone}. Then, for $r_{1,k-1}$, it suffices to consider one representative of each cyclotomic coset.
\end{corollary}

\begin{proof}
Consider any minor of $M$. Squaring all coefficients in $M$ will not change the values of  $r_{0,i},\; i=1,\ldots,k$ or $r_{1,k}$. Thus if there is a $k$-superregular matrix with $r_{1,k-1} = v$, then there is also a $k$-superregular matrix with $r_{1,k-1} = v^2$.
\end{proof}

The search can be simplified (by constant factors $O((2^m-1)^n  \cdot (2^m-1) \cdot (n-2)!$) by use of Lemma~\ref{lem:constantone}, Corollary~\ref{cor:roneisone}, and Lemma~\ref{lem:ordered},  respectively. Corollary~\ref{cor:cyclotomic} reduces complexity by an extra factor of approximately $\log_2(2^m)=m$, but this reduction is not entirely independent of the other reductions. In summary, the search algorithm is highly exponential in complexity, but the tricks allow a deeper search than would otherwise be possible.

The search algorithm is sketched in Algorithm~\ref{alg:search}. The trickier steps are explained in some detail in Remark~\ref{rem:algorithm}.


\begin{algorithm} 
\SetKwData{Left}{left}\SetKwData{This}{this}\SetKwData{Up}{up}
\SetKwFunction{Union}{Union}\SetKwFunction{FindCompress}{FindCompress}
\SetKwInOut{Input}{Input}\SetKwInOut{Output}{output}
\SetAlgoLined
\KwResult{finds good $2^m$-ary MDS codes of rate $(n-1)/n$}
\Input{Field size $2^m$, target distance ${\cal D}^*$, code length $n$}
\KwData{$\rho$ points to current position} 
\BlankLine
initialization\;
$value(r_{0,i}) := 1, i=1,\ldots,k,  value(r_{1,k})=1$\;
\emph{Precompute the set of proper submatrices} ${\cal M} = \bigcup_{\rho=r_{1,k-1}}^{r_{{\cal D}^*-2,1}} {\cal M}_\rho$\;
$\rho:=r_{1,k-1}$\;
\emph{Precompute the set of legal values} ${\cal L}(\rho)$\;
\While{ $\rho \leq r_{{\cal D}^*-2,1}$ and \emph{more coefficient values to check for $\rho$ } }{
\eIf{more coefficient values to check for $\rho$}{ 
    assign next value to coefficient at $\rho$\;
    update determinants needed for ${\cal M}_{\rho+1}$, and ${\cal L}(\rho+1)$\;
    \If{deepest level so far}{ 
        record selected values of coefficients\;}
    $\rho = \rho+1$\;
}
{
    $\rho = \rho-1$\;
    }
    }
\caption{A computer search algorithm}
\label{alg:search}
\end{algorithm}

\begin{remark}
\label{rem:algorithm}
Here we explain the steps of  Algorithm~\ref{alg:search}.
\begin{enumerate}[(i)]
  \item In essence, the algorithm runs through a search tree. At each depth of the tree, $\rho$ points to one of the variables $r_{i,j}$ in (\ref{eq:ksuperregular}). Abusing notation,  we will also say that $\rho$ points to the current depth. Throughout the course of the algorithm, $\rho$  goes back and forth along $r_{1,k-1},\ldots,r_{1,1},r_{2,k},\ldots,r_{2,1},r_{3,k},\ldots$, \emph{i. e.} along the values of the last row of (\ref{eq:ksuperregular}) in reverse order starting at $r_{1,k-1}$ (since $r_{1,k},r_{0,1},\ldots,r_{0,k}$ can be assumed to be all equal to 1 by Lemma~\ref{lem:constantone}  and Corollary~\ref{cor:roneisone}.) We will in this context use the ordering ``$<$'' to refer to the reverse order of the last row of (\ref{eq:ksuperregular}), and addition and subtraction on $\rho$ moves $\rho$ left and right, respectively, on this row.
  \item Line 3: Let $\rho$ refer to one element $r_{i,j}$ in the last row of (\ref{eq:ksuperregular}). Then ${\cal M}_\rho$ is the set of (formal) proper submatrices of (\ref{eq:ksuperregular}) which have $\rho$ in its left lower corner. If all matrices in  ${\cal M} = \bigcup_{\rho=r_{0,k}}^{r_{D',1}} {\cal M}_\rho$  for some $D' \leq D$ are nonsingular, where $D = {\cal D}^*-2$ is the target maximum degree,  then the submatrix of (\ref{eq:ksuperregular}) with $r_{D',1}$ in the lower left corner is $k$-superregular. The set ${\cal M}_\rho$ can be found by recursion. (The number of proper submatrices ${\cal M}$ is related to the Catalan numbers. We omit the details.)
  \item Line 5: At depth $\rho$, values have already been assigned for each depth $\rho' < \rho$. Hence, by keeping track of subdeterminants already computed, for each determinant corresponding to a proper submatrix in ${\cal M}_\rho$, the value in $GF(2^m)$ that would make the determinant zero can be obtained in constant time. In other words, going \emph{once} through ${\cal M}_\rho$, we can identify the set ${\cal L}(\rho)$ of all illegal values for  coefficient $\rho$.
  \item Line 6: (a) A complete search version of the algorithm will successively try all values in ${\cal L}(\rho)$. For a faster but incomplete search, the algorithm may be set to skip an arbitrary subset of values in ${\cal L}(\rho)$ at each depth $\rho$. (b) The target distance ${\cal D}^*$ is an input parameter for the algorithm. In order to determine that a code has a maximum distance ${\cal D}$, it is necessary to verify that a complete search version of the algorithm will pass depth $\rho = r_{D,1}$, but not depth $\rho = r_{D+1,1}$.
  \item Line 9: Using the set of values currently assigned to coefficients at all depths $\rho' \leq \rho$, compute subdeterminants that will be useful for computing determinants in ${\cal M}_{\rho+1}$, and initialize the set of legal values ${\cal L}(\rho + 1)$ for the next depth.
  \item Complexity:  The assumptions enabled by the Lemmas of this section, together with the efficient computation of the determinants, allow a deeper search than would be possible with a na\"{\i}ve search. However, the depth of the search tree that finds a code of degree $D$ is $(n-1)D$, and for many of the ``early depths'', a complete search needs to go through almost $2^m$ values. The size of the set of proper submatrices also grows exponentially with $n$ and $D$. So the overall complexity is at least $O(2^{mn\cdot ({\cal D}^*-2)} \cdot |{\cal M}| \cdot w_d)$ where $|{\cal M}|$  is the number of proper submatrices. 
\end{enumerate}
\end{remark}

\subsection{Codes found by computer search with ${\cal D} \geq 5$ }
\label{sec:searchresults}

Here we present codes found from computer search, for field sizes of characteristic 2 ranging from 8 to 16384, and free distances ${\cal D}\geq 5$. Exact values of $\Delta(q^m,q^m)=3$ and $\Delta(q^m,q^{m-1})=4$ are provided by Propositions~\ref{thm:d3} and \ref{thm:d4}.

In Tables~\ref{tab:bounds3}--\ref{tab:bounds14}, each row summarizes what we have discovered about rate $k/n = (n-1)/n$ codes. The $\Delta$ column lists the maximum value of ${\cal D}$ for which we have found a code with CDP $[2,3,\ldots,{\cal D}]$. The absence of a $\geq$ sign in this column indicates that we have established, through an exhaustive search, that this value of ${\cal D}$ is indeed maximum for this rate and field size.  The \emph{Coefficients} column presents one encoder that possesses this CDP, in terms of $\log_\alpha()$ of the coefficients $(r_{1,k},\ldots,r_{1,1}),(r_{2,k},\ldots,r_{2,1}),(r_{3,k},\ldots,r_{3,1}),\ldots$, where $\alpha$ is the primitive element of the field. Note that the degree zero terms, $(r_{0,1},\ldots,r_{0,k})$ are suppressed since they are  assumed to be identically $1 = \alpha^0$. The \emph{R} column contains the rareness of the code, which will be  explained in Section~\ref{subsec:rareness}. In the \emph{Reference} column we include references in the few cases where ``similar codes'' (\emph{i. e.} codes over the same field which have the same CDP, but that do not necessarily possess a systematic encoder) have previously been described in the literature.
We do not list encoders found by the search if we also found codes with the same set of parameters (rate, CDP) over a smaller field.
Also, due to Lemma~\ref{lem:shorten}, we do not list codes of rate $(n-2)/(n-1)$ if there exist codes of rate  $(n-1)/n$ with the same CDP:

\begin{lemma}
\label{lem:shorten}
If a systematic MDS code with free distance $\cal D$ and rate $k/(k+1)$ exists for $k>1$,  then there is also a systematic MDS code with free distance $\cal D$ and rate $(k-1)/k$.
\end{lemma}

\begin{proof}
Shorten the $k$-superregular matrix $H'^{(D)}$ by selecting any $j_0\in\{1,\ldots , k\}$ and removing columns $j_0, j_0+k, j_0+2k, \ldots $ in $H'^{(D)}$.

\end{proof}

%
%
%
%
%
%
%
%


\begin{table}
\begin{center}

\begin{tabular}{|r|r|c|c|c|}

\hline
  $n$ & $\Delta$  & Coefficients & R & Remark\\
 \hline
 2 & 6 & 0, 1, 4, 3 & 0.035 &  \cite{JustesenHughes1974}\\  %
  \hline
\end{tabular}

\vspace{0.1cm} \caption{Table of bounds on $\Delta(2^3,n)$ for the field defined by $1+\alpha+\alpha^{3} = 0$. 
}
\label{tab:bounds3}
\end{center}
\end{table}

\begin{table}
\begin{center}
\begin{tabular}{|r|r|c|c|c|}
\hline
  $n$ & $\Delta$  & Coefficients& R  & Remark\\
  \hline
 2 & 7 & 0, 1, 4, 3, 0 & 0.024  &   \\  %
 3 & 5 & 0 1, 4 0, 1 7  & 0.014   &   \cite{StronglyMDS2006} \\
  \hline

\end{tabular}

\vspace{0.1cm} \caption{Table of bounds on $\Delta(2^4,n)$ for the field defined by $1+\alpha+\alpha^{4} = 0$. 
Please also see Example~\ref{ex-2}. 
}
\label{tab:bounds4}
\end{center}
\end{table}

\begin{table}
\begin{center}
\begin{tabular}{|r|r|c|c|}
\hline
  $n$ & $\Delta$  & Coefficients& R \\
  \hline
 2 & 9 & 0, 1, 19, 5, 24, 15, 0 & $3.4 \cdot 10^{-8}$    \\  %
  3 & 6 & 0 1, 11 28, 21 6, 24 11 & $4.4 \cdot 10^{-5}$ \\
  5 & 5 & 0 1 18 2, 5 8 17 25, 3 2 13 18 & $5.2 \cdot 10^{-11}$  \\

\hline
\end{tabular}

\vspace{0.1cm} \caption{Table of bounds on $\Delta(2^5,n)$ for the field defined by $1+\alpha^{2}+\alpha^{5} = 0$. 
}
\label{tab:bounds5}
\end{center}
\end{table}

\begin{table}
\begin{center}
\begin{tabular}{|r|r|c|c|}
\hline
  $n$ & $\Delta$  & Coefficients& R\\
  \hline
  2 & 10 & 0, 1, 6, 61, 60, 46, 28, 23 & $ 1.2 \cdot 10^{-10}$ \\  
   3 & 7 & 0 1, 6 0, 2 37, 21 44, 55 28 & $ 4.1 \cdot 10^{-11}$  \\
   4 & $\geq$6 & 0 1 6, 2 6 26, 13 61 38, 30 33 60 &  $ 1.4 \cdot 10^{-11}$\\ 
   7 & $\geq$5 & 0 1 6 2 12 3, 14 36 26 25 51 13, 19 60 16 62 5 58 & $ 3.2 \cdot 10^{-20}$\\ 
\hline

\end{tabular}

\vspace{0.1cm} \caption{Table of bounds on $\Delta(2^6,n)$ for the field defined by $1+\alpha+\alpha^{6} = 0$. 
}
\label{tab:bounds6}
\end{center}
\end{table}

\begin{table}
\begin{center}
\begin{tabular}{|r|r|c|c|}

 \hline
   $n$ & $\Delta$  & Coefficients& R\\
  \hline
5 & $\geq$6 & 0 1 31 2, 62 103 64 125, 51 57 19 110, 11 39 43 114 &  $8 \cdot 10^{-18}$  \\ 
 8 & $\geq$5 & 0 1 31 2 62 32 103, 3 31 15 0 7 1 63, &   $6.4 \cdot 10^{-16}$ \\   
& & 8 94 119 51 41 10 17 &  \\
\hline

\end{tabular}

\vspace{0.1cm} \caption{Table of bounds on $\Delta(2^7,n)$ for the field defined by $1+\alpha^{3}+\alpha^{7} = 0$. 
}
\label{tab:bounds7}
\end{center}
\end{table}

\begin{table}
\begin{center}
\begin{tabular}{|r|r|c|c|}
\hline
  $n$ & $\Delta$  & Coefficients& R\\
  \hline
  2 & $\geq$11 & 0, 1, 25, 3, 0, 198, 152, 56, 68 & $2.2 \cdot 10^{-7}$ \\
   3 & $\geq$8 & 0 1, 25 0, 1 238, 100 106, 195 245, 37 33 & $2.0 \cdot 10^{-12}$ \\ 
  4 & $\geq$7 & 0 1 25, 2 25 198, 1 14 228, 113 74 214, 21 250 172 & $\approx 2 \cdot 10^{-17}$ \\  
   11 & $\geq$5 &0 96 95 176 156 169 160 81 11 245, 107 5 223 167 7 177 98 238 93 53, & $ \approx 3 \cdot 10^{-28}$ \\ 

  & &  37 208 233 89 75 74 184 31 119 100 &  \\
\hline

\end{tabular}

\vspace{0.1cm} \caption{Table of bounds on $\Delta(2^8,n)$ for the field defined by $1+\alpha^{2}+\alpha^{3}+\alpha^{4}+\alpha^{8} = 0$. 
}
\label{tab:bounds8}
\end{center}
\end{table}

\begin{table}
\begin{center}
\begin{tabular}{|r|r|c|c|}
\hline
  $n$ & $\Delta$  & Coefficients& R\\
  \hline
 2 & $\geq$12 & 0,    54,  91, 181, 267, 291, 379, 28, 95, 143 & $ 1.4 \cdot 10^{-11}$  \\ 
   6 & $\geq$6 & 0 280 362 276 426, 206 155 326 324 360, & $ 3.9 \cdot 10^{-11}$ \\ 
& & 356 447 507 312 144, 224 375 236 55 448 & \\
  13 & $\geq$5 & 0 19 325 321 356 397 317 455 98 130 149 413, &  $ 8.4 \cdot 10^{-27}$\\ 
& & 48 101 120 272 209 188 405 352 46 343 289 152, &  \\
 & &  318 80 256 98 255 274 147 340 392 453 30 451 & \\
 \hline

\end{tabular}

\vspace{0.1cm} \caption{Table of bounds on $\Delta(2^9,n)$ for the field defined by $1+\alpha^{4}+\alpha^{9} = 0$. 
}
\label{tab:bounds9}
\end{center}
\end{table}

\begin{table}
\begin{center}
\begin{tabular}{|r|r|c|c|}
\hline
  $n$ & $\Delta$  & Coefficients& R\\
  \hline
3 & $\geq$9 & 0 603, 246 106, 115 693, 483 544, 603 152, 815 788, 984 721 & $\approx  10^{-15}$ \\ 
 5 & $\geq$7 & 0 498 997 964, 560 214 101 723, 453 111 370 54,& $5 \cdot 10^{-18}$  \\ 
 & & 455 17 625 509, 904 431 926 856 & \\
  8 & $\geq$6 &  0 322 804 12 140 1004 384, 778 916 786 247 586 698 294, & $3 \cdot 10^{-24}$ \\ 
 & & 379 7 784 239 817 284 398, 178 588 110 41 425 976 393 & \\
  17 & $\geq$5 & 0 1 77 2 154 78 956 3 10 155 325 79 618 957 231 4, & $4 \cdot 10^{-39}$ \\ 
& &   308 0 4 77 11 1 200 10 80 3 24 155 87 325 619 618, &  \\
& & 958 768 255 404 577 976 368 374 709 33 530 109 677 594 652 226 &  \\
\hline

\end{tabular}

\vspace{0.1cm} \caption{Table of bounds on $\Delta(2^{10},n)$ for the field defined by $1+\alpha^{3}+\alpha^{10}=0$. 
}
\label{tab:bounds10}
\end{center}
\end{table}

\begin{table}
\begin{center}
\begin{tabular}{|r|r|c|c|}

\hline
  $n$ & $\Delta$  & Coefficients& R\\
  \hline
 2 & $\geq$13 & 0, 1992, 813, 1890, 440, 630, 1947, 1574, 1356, 234, 1266 & $ 1.0 \cdot 10^{-9}$  \\  
 4 & $\geq$8 & 0 1809 1118, 2027 1610 539, 1042 7 1730, &  $5.6 \cdot 10^{-15}$\\ 
& &  2020 591 1459, 902 899 1584, 172 1192 513 &  \\
  9 & $\geq$6 &0 1999 762 1845 1102 1115 1014 328, 1349 345 498 1561 27 987 1300 1793,  & $2.0 \cdot 10^{-22}$ \\ 
 & &  1728 562 488 304 43 71 1911 1140, 1524 660 465 327 322 748 1574 1414 &  \\
 \hline

\end{tabular}

\vspace{0.1cm} \caption{Table of bounds on $\Delta(2^{11},n)$ for the field defined by $1+\alpha^{2}+\alpha^{11} = 0$. 
}
\label{tab:bounds11}
\end{center}
\end{table}

\begin{table}
\begin{center}
\begin{tabular}{|r|r|c|c|}
\hline
  $n$ & $\Delta$  & Coefficients& R\\
  \hline
   2 & $\geq$14 & 0, 3294, 1040, 448, 3624, 2406, 826, 1122, 587, 1034, 342, 4037 & $<10^{-15}$  \\
     6 & $\geq$7 & 0 3202 2711 92 2688, 3908 1649 1252 3897 1604, 3687 3602 1603 2339 1350,   & $1.2 \cdot 10^{-14}$ \\ 
& & 1700 2969 104 3406 2679, 1345 919 3302 2116 810 & \\
  11 & $\geq$6 &0 669 4050 4007 745 3863 324 1617 3951 1343, & $3 \cdot 10^{-31}$ \\  
  & & 703 1123 782 3343 1919 3177 1839 1006 2183 426,  &  \\
& & 2139 2050 1676 1187 3222 467 1764 2387 2868 641, &  \\
  & & 2564 2249 3187 3114 3228 743 443 1220 3540 2620 &  \\
  \hline

\end{tabular}

\vspace{0.1cm} \caption{Table of bounds on $\Delta(2^{12},n)$ for the field defined by $1+\alpha^{3}+\alpha^{4}+\alpha^{7}+\alpha^{12} = 0$. 
}
\label{tab:bounds12}
\end{center}
\end{table}

\begin{table}
\begin{center}
\begin{tabular}{|r|r|c|c|}

\hline
  $n$ & $\Delta$  & Coefficients& R\\
  \hline
 3 & $\geq$10 & 0 337, 7672 6843, 3625 3361, 7970 7490,  &  $3.6 \cdot 10^{-11}$\\ 
 &   &  5531 2322, 5227 5758, 133 2290, 1453 189  &  \\
    5 & $\geq$8 & 0 441 2192 3413, 3222 7502 7405 4155, 88 5939 343 6171,  & $\approx 5 \cdot 10^{-21}$ \\ 
& & 1082 8149 2823 7269, 8022 6454 4999 3373, 3518 442 710 6968 & \\
  7 & $\geq$7 & 0 5160 5711 7681 748 5319, 2131 6233 723 4539 7315 5654,    & $2 \cdot 10^{-19}$ \\ 
 & & 5126 7465 3577 6826 5553 1131, 4954 6763 6593 1568 7157 8112,  & \\
 & & 1961 4310 877 2927 7197 2672 & \\
  13 & $\geq$6 &  0 5645 7651 3109 2678 802 6934 1946 5589 2833 5821 38,  &  $\approx 8 \cdot 10^{-37}$\\ 
& & 5394 2500 5877 3141 4724 3374 5191 7218 4844 423 822 6875,  &  \\
 & &  5712 6619 3935 6414 8025 1422 4391 5698 5481 6850 2635 4786,  &  \\
& &  556 2558 1063 5172 566 7978 3664 5848 3859 6905 6434 71 &  \\
\hline

\end{tabular}

\vspace{0.1cm} \caption{Table of bounds on $\Delta(2^{13},n)$ for the field defined by $1+\alpha+\alpha^{3}+\alpha^{4}+\alpha^{13} = 0$. 
}
\label{tab:bounds13}
\end{center}
\end{table}

\begin{table}
\begin{center}
\begin{tabular}{|r|r|c|c|}

\hline
  $n$ & $\Delta$  & Coefficients& R\\
  \hline
 4 & $\geq$9 &  0 61 9533, 1260 4487 6469, 3689 8777 4510, 11257 13252 1239,  & $3 \cdot 10^{-14}$ \\ 
& &  15121 10306 11679, 9618 13110 4549, 12420 5210 13006 &  \\
   8 & $\geq$7 & 0 14132 6404 8841 7620 6707 1150,     &  $1.4 \cdot 10^{-22}$ \\ 
& & 14939 8238 9174 9560 1677 4156 11112, & \\
 & & 11424 2037 7827 4640 11071 14007 6628,  & \\
 & &  13374 10684 2080 14648 1097 14383 1198, & \\
 & &  10966 15875 9746 9595 13007 4019 1354 & \\
  15 & $\geq$6 &  0 15439 10581 4136 503 11096 5590 8608 16006 8229 562 15423 14311 16137,   & $2 \cdot 10^{-38}$ \\ 
& & 5899 1875 8985 16334 15293 13429 5172 5303 9128 109 10068 1358 7752 6288,   &  \\
& &  13251 13386 11513 2438 443 15582 4641 2845 3509 12593 6608 14686 11470 15578,   &  \\
& &  8683 12489 444 8891 4727 12844 12383 5530 4478 9079 9226 5886 6790 8363 &  \\
\hline

\end{tabular}

\vspace{0.1cm} \caption{Table of bounds on $\Delta(2^{14},n)$ for the field defined by $1+\alpha+\alpha^{11}+\alpha^{12}+\alpha^{14} = 0$. 
}
\label{tab:bounds14}
\end{center}
\end{table}

\begin{example}
\label{ex-2}
According to Table~\ref{tab:bounds4}, for the finite field $GF(2^4)$ defined by $1+\alpha+\alpha^{4}= 0$, there exists a systematic code of rate 2/3 and with  CDP$=[2,3,4,5]$. An example of such a code is represented by  $(r_{1,2},r_{1,1}),(r_{2,2},r_{2,1}),(r_{3,2},r_{3,1}) = $ $(\alpha^0,\alpha^1),(\alpha^4,\alpha^0),(\alpha^1,\alpha^7) = $ $(1,\alpha),(\alpha^4,1),(\alpha,\alpha^7)$, and (implicitly)  $(r_{0,1},\ldots,r_{0,k}) = $$(1,\ldots,1)$. Thus the code has a polynomial parity check matrix 
\[H(x) = (1+\alpha x + x^2+ \alpha^7 x^3,1+x+\alpha^4 x^2+ \alpha x^3,1)\] and encoder/generator matrix  
\[
G(x) = \begin{pmatrix}
1 & 0 & 1+\alpha x + x^2+ \alpha^7 x^3 \\
0 & 1 & 1+x+\alpha^4 x^2+ \alpha x^3
 \end{pmatrix}
\]
Obviously, $G(x)H(x)^\top = (0,0)$. The absence of a $\geq$ symbol in the $\Delta$ column in Table~\ref{tab:bounds4} indicates that a complete search of all systematic MDS codes of rate 2/3 reveals that ${\cal D}=5$ is maximum. 
The $R$ column, as will be explained later, indicates that one in seventy random assignments of nonzero values for $(r_{1,1},r_{1,2}),(r_{2,1},r_{2,2}),(r_{3,1},r_{3,2})$ will give a code with the same CDP, \emph{i. e.} codes with these parameters are not very rare. A \emph{nonsystematic} code over $GF(2^4)$ with degree $\delta=2$ and CDP$=[2,3,4,5]$ was presented in \cite{StronglyMDS2006}.
\end{example}

\section{Upper bounds and code assessment}
\label{sec:bounds}

It would be useful to determine upper bounds on $\Delta(q^m,n)$ in order to assess how good the codes from random search are with respect to optimum. The Heller bound \cite{Heller68} relates convolutional codes with a given free distance ${\cal D}$ with its truncated block codes, and uses known bounds on block codes to determine convolutional code parameters that cannot be achieved. Unfortunately the Heller bound is of limited use in our case, since the truncated code will actually have a much lower minimum distance than $\cal D$ when viewed as a block code, and also since exact bounds on block codes in the range of parameters that we are interested in here are not well known. Moreover, the approach of sphere packing for binary codes \cite{RosnesYtrehus2004} cannot be easily adapted to the current case, since the structure of optimum nonbinary codes turns out to be quite different from that of optimum binary codes\footnote{Optimum binary convolutional codes tend to require parity check matrices with many $r_{1,j}=0$, whereas we have seen that in the nonbinary case, all degree one coefficients $r_{1,j}$ are nonzero. These differences impose different combinatorial constraints in the binary and the nonbinary case.}.

A simple bound is described in the next subsection. In Subsection~\ref{subsec:rareness} we present an alternative way of describing how great our codes are, through the concept of \emph{rareness}.

\subsection{A simple bound}
The following simple bound is tight for ${\cal D} \leq 4$.

\begin{theorem}
\label{th-upb-1}
For rate $(n-1)/n$ codes over $GF(q^m)$ with $CDP = [2,3,\dots,{\cal D}],$ $n-1\leq (q^m-1)/({\cal D}-2)$.
\end{theorem}

\begin{proof}

For ${\cal D} = 3$ the result follows from Proposition~\ref{thm:d3}. Assume that ${\cal D} = 4$. Recall that all coefficients are nonzero. Consider the $2 \times 2$ minors of type
\begin{equation}
\begin{vmatrix}
\label{}
1 & 1\\
r_{1,s} & r_{1,t}
 \end{vmatrix} = r_{1,s} + r_{1,t},
 \begin{vmatrix}
\label{}
r_{1,s} & r_{1,t}\\
r_{2,s} & r_{2,t}
 \end{vmatrix}=r_{1,s}r_{2,t}+r_{1,t}r_{2,s} , \mbox{ and }
  \begin{vmatrix}
\label{}
r_{1,s} &1\\
r_{2,s} & r_{1,t}
 \end{vmatrix} = r_{2,s}+r_{1,s}r_{1,t}.
\end{equation}

 From the conditions on the $2 \times 2$ proper minors, since all those minors have to be nonzero,  it follows that in order to have ${\cal D} > 3$, the values in the sets $\{r_{1,1},\ldots,r_{1,k}\}$ and $\{r_{2,1}/r_{1,1},\ldots,r_{2,k}/r_{1,k}\}$ must be $2k$ distinct values in $GF(q^m) \setminus \{0\}$.

Now consider a code with ${\cal D} > 4$. Then the minors

\begin{equation}
\begin{vmatrix}
\label{}
r_{2,s} & r_{2,t}\\
r_{3,s} & r_{3,t}
 \end{vmatrix} = r_{2,s}r_{3,t} + r_{2,t}r_{3,s},
 \begin{vmatrix}
\label{}
r_{2,s} & 1\\
r_{3,s} & r_{1,t}
 \end{vmatrix}=r_{2,s}r_{1,t}+r_{3,s}, \mbox{ and }
  \begin{vmatrix}
\label{}
r_{2,s} & r_{1,t}\\
r_{3,s} & r_{2,t}
 \end{vmatrix} = r_{2,s}r_{2,t}+r_{1,t}r_{3,s}.
\end{equation}

Again they all have to be nonzero, and this implies that the set $\{r_{3,1}/r_{2,1}, \ldots , r_{3,k}/r_{2,k}\}$ is a new set of $k$ different values, and they are all different from the values in the sets $\{r_{1,1},\ldots,r_{1,k}\}$ and $\{r_{2,1}/r_{1,1},\ldots,r_{2,k}/r_{1,k}\}$. So in order to have ${\cal D}\geq 5$ we need to have at least $3k$ different non zero elements in the field.

 Generalizing the argument, it follows  that all   $r_{i,t}/r_{i-1,t}$ for $1\leq i\leq {\cal D}-2, 1 \leq t \leq k$ are  distinct nonzero values.
\end{proof}


\subsection{Rareness}
\label{subsec:rareness}

In this section we address the probability that a randomly generated convolutional code over $GF(2^m)$ of rate $(n-1)/n$ will be an MDS code with CDP of $[2,\ldots,{\cal D}]$. By ``randomly generated'' code we will mean one generated by a random systematic encoder, where each coding coefficient $r_{i,j}$  is selected independently and uniformly in $GF(2^m)\setminus\{0\}$. We define this probability as the \emph{rareness} of the parameter pair $(n,{\cal D})$.

For small values of $n$ and ${\cal D}$, the exact value of the rareness can be determined as a by-product of a complete code search. Since for large parameters it quickly becomes intractable to determine the best codes, it also quickly turns difficult to compute exact results for rareness. However, it is possible to obtain estimates of rareness, as described below.

First assume that  a complete search is applied. This will determine the set  ${\cal G}(\rho^*,n,m)$ of distinct sequences $r_{1,k-1}, \ldots, r_{1,1},$ $r_{2,k}, \ldots, \rho^*$  over $GF(2^m)$  for which all proper submatrices in ${\cal M}_{\rho'}$, $\rho' \leq \rho^*$ are nonsingular. Thus the probability that a given randomly selected sequence corresponds to a path in the search tree that satisfies the conditions at depth  $\rho^*$ is
\[P_R(\rho^*,n,m) = \frac{|{\cal G}(\rho^*,n,m)|}{(2^m-1)^{|\rho^*|}}.\]

For $|\rho^*|>1$, define
\begin{equation}\label{eq:Prhocond}
  P_R(\rho^*,n,m |  \rho^*-1) = \frac{P_R(\rho^*,n,m)}{P_R(\rho^*-1,n,m)} = Avg(\frac{{|\cal L}_{\rho^*}|}{2^m-1})
\end{equation}
where $Avg()$ is the average computed over the complete search. $P_R(\rho^*,n,m |  \rho^*-1)$ is the average conditional probability that a random generator which satisfies depth $\rho^*-1$ in the search tree also satisfies depth $\rho^*$.
For large parameters we are not able to carry out a complete search. However, we can perform deep but incomplete searches, which also provide estimates of the conditional probabilities $P_R(\rho^*,n,m |  \rho^*-1)$ in (\ref{eq:Prhocond}) as . These estimates will be quite accurate especially for the first depths, and hence they can be changed together to obtain an estimate for $P_R(\rho^*,n,m)$. As long as there is a substantial number of different search tree paths leading to depth $\rho^*-1$, the estimate $P_R(\rho^*,n,m |  \rho^*-1)$ should be reasonably good.  Hence we can also estimate $P_R(\rho^*,n,m|  \rho^*-1))$ as

\[\tilde{P}_R(\rho^*,n,m |  \rho^*-1) ) = \tilde{Avg}(\frac{{|\cal L}_{\rho^*}|}{2^m-1}) \]
where $\tilde{Avg}()$ is the (weighted) average computed over the incomplete search, and we can then estimate $P_R(\rho^*,n,m)$ as
\[\tilde{P}_R(\rho^*,n,m  ) = \prod_{\rho=(1,k-1)}^{\rho^*}\tilde{P}_R(\rho-1,n,m|  \rho-1)\]
where for $\rho = (1,k-1)$, $\tilde{P}_R(\rho-1,n,m|  \rho-1)=1$.

In Tables~\ref{tab:bounds3}--\ref{tab:bounds14}, we include the exact rareness in cases where we can perform a complete search, and otherwise we include the estimate. We concede that this approach is not foolproof. For example,  the construction in Proposition~\ref{thm:d4}, is unique at least for field sizes up to 32. For other choices for the first layer of coefficients $r_{1,1},\ldots,r_{1,2^{m-1}}$ than indicated in the proof of  Proposition~\ref{thm:d4}, it appears that the search tree ends up being considerably shallower. The rareness of the construction in Proposition~\ref{thm:d4}, \emph{i. e.} the probability that a random sequence will match  that construction exactly, is $2^{m-1}  \cdot (2^{m-1}-1)!/(2^m-1)^{2^m-3}$. Already for $m=5$ the rareness is about $10^{-30}$, for $m=8$ less than  $10^{-393}$. Hence, if for an arbitrary set of search parameters  there exists a very rare construction that is not caught by the incomplete search, the estimates for the deepest values of $\rho^*$ may be unprecise. However, we do believe that our estimates of $P_R(\rho^*,n,m)$  provide some intuition about the difficulty of reaching a certain depth in the search tree with a random path, and in the cases where we are able to carry out a complete search, we also note that the estimates as described here are pretty accurate with a modest non-exhaustive search effort.

Figure~\ref{fig-logprobabilities64} contains exact values (for $n=2,3$) and estimates (for $n=4,7$) of $P_R(\rho^*,n,6)$. Please see the figure caption for explanations.  We have also include rareness estimates in Tables~\ref{tab:bounds3}--\ref{tab:bounds14}.


\begin{figure}
  \centering
  \includegraphics[width=18cm]{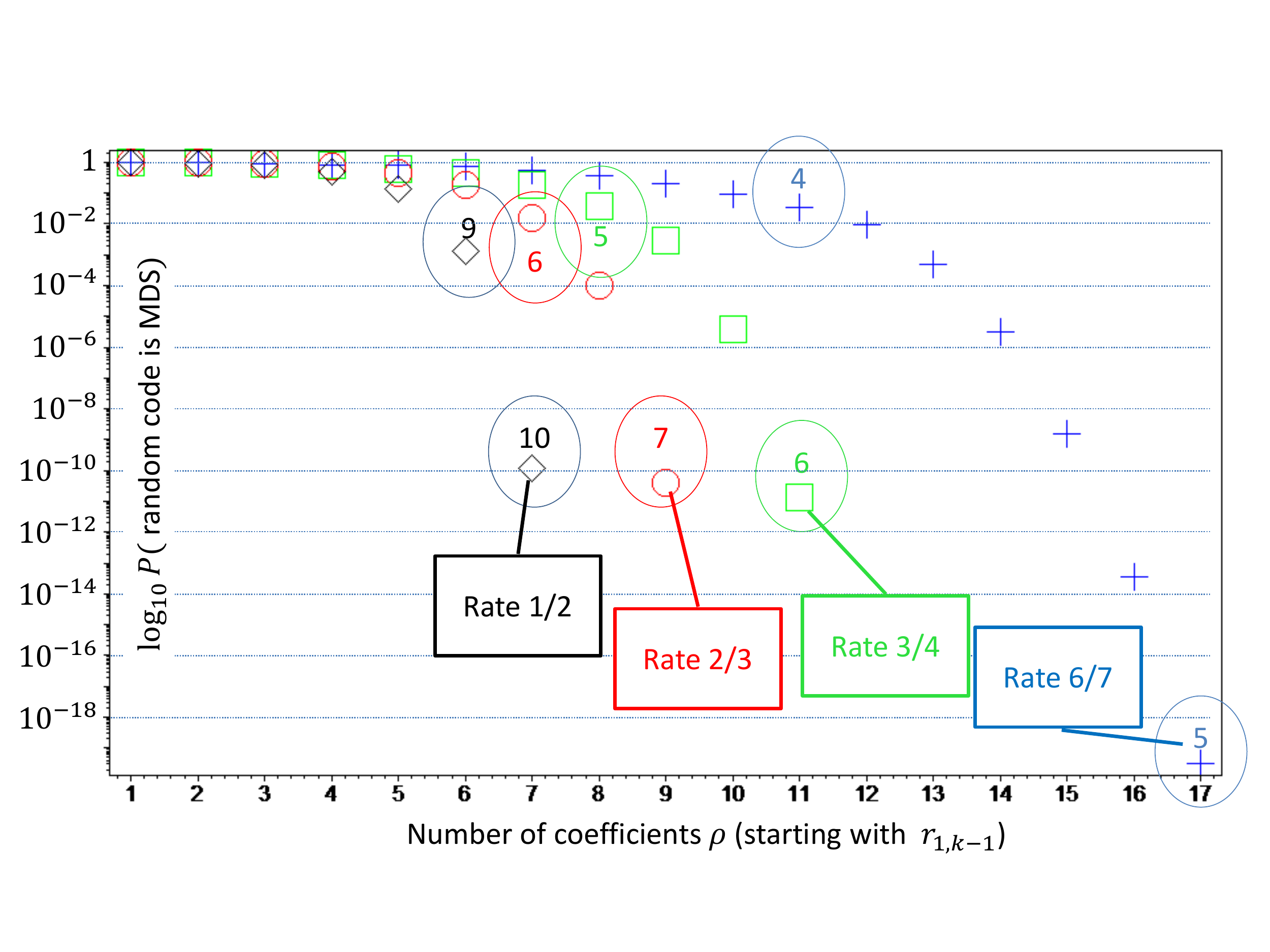}
  \vspace{-1cm}
  \caption{Rareness $P_R(\rho,n,6)$ of codes for $GF(64)$ for $n \in \{2,3,4,7\}$: Exact rareness $P_R(\rho,n,6)$ for $\rho \leq 7$, estimates $\tilde{P}_R(\rho,n,6)$ for $n>7$. In the figure, the search depth $\rho$ is measured in terms of number of coefficients. In order to construct a rate 6/7 encoder of distance  ${\cal D}=5$, it is necessary to find a sequence of 17 coefficients $r_{1,5}, \ldots, r_{1,1},r_{2,6},\ldots,r_{3,1}$. To get an encoder with distance ${\cal D}=4$, it suffices with 11 coefficients. Similar for the other cases.  }\label{fig-logprobabilities64}
\end{figure}

\section{Conclusion and open problems}

Motivated by the practical problem of fast recovery of a coded packet-erasure channel, we have studied systematic MDS convolutional codes over $GF(2^m)$. We have characterized them in terms of $k$-superregularity of a certain matrix. We have presented new optimum constructions for free distances ${\cal D}\leq 4$, tables of new codes found by computer search, and a combinatorial upper bound which is tight in the case of small free distances. In order to assess how ``good'' a code is, we have also introduced the concept of \emph{rareness}.

It would be interesting to establish upper bounds that are tight also for larger free distances. Another issue would be to study whether there exist general algebraic constructions, similar to the one in  Proposition~\ref{thm:d4}, for systematic MDS codes of free distance ${\cal D}\geq 5$.

It would also be of some theoretical interest to optimize the CDP of \emph{strongly}-MDS codes over $GF(2^m)$ under an additional constraint on the degree $\delta$ of their minimal encoders. We have not considered this problem since the complexity of Viterbi decoding of such codes is prohibitive for all but small values of the product $m \cdot \delta$ (and since it seems difficult). 


\section*{Appendix A : Proof of Lemma~\ref{lemma:biglemma}}

\begin{proof}

Before starting we will set some notations.

In $H^{(D)}$ for each $s=1,\ldots, (D+1)$ let $C_s$ be the set of column indices $C_s=\{(s-1)(k+1)+1, (s-1)(k+1)+2, \ldots, s(k+1)\}$.

Taking into account the way $H^{(D)}$ is constructed it is clear that for any $s\leq D$, the submatrix of $H^{(D)}$  formed by the first  $s+1$ rows and the columns in $C_1, \ldots, C_{s+1}$ is $H^{(s)}$ (the parity-check matrix of the $s$-th truncation).

Also the submatrix formed by the last $s+1$ rows and the columns in $C_{D-s+1}, \ldots, C_{D+1}$ is $H^{(s)}$.

For each set of column indices $C_s$ the last index is $s(k+1)$ and the corresponding column in $H^{(D)}$ is the $s$-th column of the identity $I_{D+1}$.

In an analogous way we will call $C'_s$ the set of column indices  $C'_s=\{(s-1)k+1, (s-1)k+2, \ldots, sk\}$ in the matrix $H'^{(D)}$.

In what follows we will use the same name for a square submatrix and for the corresponding minor since it will create no confusion.

Now we start the proof.

\begin{enumerate}

\item Assume that the CDP is $(d_0=2, d_1=3,\ldots ,d_D= D+2)$.

 In particular $d_0=2$ implies that all the entries $r_{0,j}$ for $j\in C_1$ are non zero.

  Let $M'$ be a proper minor of $H'^{(D)}$ of size $p\times p$ formed by the  entries of $H'^{(D)}$ in rows with indices $1\leq i'_1 < i'_2<\cdots < i'_p\leq s+1$ and columns with indices $1\leq j'_1< j'_2< \cdots < j'_p\leq k(D+1)$.

    Since $M'$ is proper we have $j'_l\leq ki'_l$ for $l=1,\ldots, p$.

    Let $F'=\{i'_1, \ldots, i'_p\}$ be the set of row indices in $M'$.
    From $M'$ we construct a $(D+1)\times (D+1)$ minor $M$ in $H^{(D)}$ by doing the following:

    \begin{itemize}
    \item The row indices are $\{1,\ldots , D+1\}$.
    \item For each $s \in \{1,\ldots, D+1\}$ we define the column index $j_s$ as follows:
    \begin{itemize}
    \item If $s\in F'$ then there exists a unique $l(s)\in \{1,\ldots , p\}$ such that $i'_{l(s)}=s$ (note that $l(s)\leq s$). Considering the corresponding column index in $M'$ we have $j'_{l(s)}=q_{l(s)}k+r_{l(s)}$ with $0\leq q_{l(s)}\leq D$ and $1\leq r_{l(s)}\leq k$ unique. (We note that $l$ is an increasing function of $s$ and also that $j'_{l(s)}\leq ki'_{l(s)}=ks$, which implies $q_{l(s)}<s$).

        Then define $j_s=q_{l(s)}(k+1)+r_{l(s)}$.
        Clearly $j'_{l(s)}\in C'_{q_{l(s)}+1}$ and $j_{l(s)}\in C_{q_{l(s)}+1}$ and actually the corresponding columns are identical.
        \item If $s\notin F'$ then $j_s=s(k+1)$, so the corresponding column is the last in the block with column indices $C_s$.

    \end{itemize}

    \end{itemize}

   Let us note that $1\leq j_1, \ldots, j_{D+1}\leq (D+1)(k+1)$ but those column indices are not guaranteed to be ordered in increasing order as the $j'_s$ were.

 The added columns will form a submatrix which is $I_{D+1-p}$ in the rows that were not in $F'$, and we have
 \[M=\left( \begin{array}{c|c}
 I_{D+1-p} & \star \\
 \hline
 0 & M' \end{array}\right)\]
 therefore the value of minor $M'$ is the same as the value of $M$ and in order to see that $H'{(D)}$ is $k$-superregular we just need to check that $M\neq 0$. We will proceed in a recursive way using that each $s$-th truncation will provide minimum distance $d_s=s+2$ for each $s\leq D$.

 \begin{itemize}

   \item $M$ has at least one column index in $C_1$.

   Proof: If there are no columns in $C_1$ it means that $1\in F'$ (otherwise $j_1=1(k+1)$, which is in $C_1$,  would be in $M$).
   $1\in F'$ implies $i'_1=1$, hence $l(1)=1$ and $j'_1\leq k i'_1=k$, that is, $j'_1=0\cdot k+r_1$, and $j_1=0(k+1)+r_1<(k+1)$. This means $j_1\in C_1$ which would contradict the assumption.

   \item If $M$ has exactly 1 column in $C_1$ then the other $D$ columns have indices in $C_2 \cup \cdots \cup C_{D+1}$ and all have 0 in the first position, so we have

       \[M=\left( \begin{array}{c|c}
       r_{0,j_1} & 0 \cdots 0 \\
       \hline
       * & M_{2\ldots D+1} \end{array}\right)\]
       where $M_{2\ldots D+1}$ is the submatrix of $M$ formed by the last $D$ rows and the last $D$ columns. Since $r_{0,j_1}\neq 0$ we have $M\neq 0$ if and only if $M_{2\ldots D+1}\neq 0$, and we can proceed working with $M_{2\ldots D+1}$ in $H^{(D-1)}$ in the same way.

   \item If $M$ has at least two columns in  $C_1$. Suppose that $s$ is the first index for which we have that at least two columns of $M$ are in $C_1$, at least 3 are in $C_1\cup C_2$, $\ldots$ , at least $s+1$ columns are in $C_1\cup C_2\cup \cdots \cup C_s$ but there are no $s+2$ columns in $C_1\cup C_2\cup \cdots \cup C_s\cup C_{s+1}$.

       This clearly implies that there are no column of $M$ in $C_{s+1}$.

       Now let us consider each $t\in \{1,\ldots , s+1\}$.
       \begin{itemize}
       \item If $t\in F'$ there exists $l(t)\leq t\leq s+1$ such that $i'_{l(t)}=t$.

       $j'_{l(t)}=q_{l(t)}k+r_{l(t)}\leq k i'_{l(t)}=kt$, therefore $q_{l(t)}\leq t-1$, and from here $j_t=q_{l(t)}(k+1)+r_{l(t)}\leq (t-1)(k+1)+r_{l(t)}$. So $j_t\in C_1 \cup C_2 \cup \cdots \cup C_t\subseteq C_1 \cup C_2 \cup \cdots \cup C_{s+1}$, but it cannot be in $C_{s+1}$, then $j_t\in C_1 \cup C_2 \cup \cdots \cup C_s$.

       \item If $t\notin F'$. Note that in this case $t\leq s$ since $s+1 \notin F'$  implies column $(s+1)(k+1)$ is in $M$ and in $C_{s+1}$, contradicting that there were no columns in $C_{s+1}$

           Then $j_t=t(k+1)\in C_1 \cup C_2 \cup \cdots \cup C_s$.

       \end{itemize}

       We have proven that even though indices $j_1, \ldots j_{s+1}$ are not ordered in increasing order, we have that they are all in $C_1 \cup C_2 \cup \cdots \cup C_s$.

       On the other hand, index $j_{s+2} \notin C_1 \cup C_2 \cup \cdots \cup C_{s+1}$.



  Hence, $M$ can be decomposed as
  \[M=\left( \begin{array} {c|c}
  M_{1\ldots s+1} & 0 \\
  \hline
  * & M_{s+2\ldots D+1}\end{array}\right)\]

  where  $M_{1\ldots s+1}$ is the part of $M$ corresponding to the first $s+1$ rows and columns and we have proven it is contained in the submatrix of $H^{(D)}$ formed by the first $s+1$ rows and the first $(s+1)(k+1)$ columns, which actually is $H^{(s)}$ and it is guaranteed to be non zero because $d_s=s+2$ and the minor satisfies the condition that it has at least 2 columns among the first $k+1$ columns of $H^{(s)}$, at least three among the first $2(k+1)$, $\ldots$, and at least $s+1$ among the first $s(k+1)$.

  Minor $M_{s+2\ldots D+1}$ is formed by the last $D-s$ rows and columns of $M$ and it is contained in the submatrix of $H^{(D)}$ formed by the last $D-s$ rows and the last $(D-s)(k+1)$ columns, which is $H^{(D-s-1)}$ and the same argument used so far can be used to  prove  that is is non zero by decomposing it further into blocks; each of them nonzero.

  Finally, we can note that $M$ will have at most one column index in $C_{D+1}$. Having at least two would imply that $M'$ has also at least two columns in $C'_{D+1}$  and this would contradict the condition of $M'$ being proper since $i'_{p-1}\leq D, j'_{p-1}\in C'_{D+1}$ implies $j'_{p-1}\geq kD+1>kD=k
  i'_{p-1}$.

    \end{itemize}

 \item Suppose now that $H'^{(D)}$ is $k$-superregular.

 Consider a minor $M$ of size $(D+1)$ in $H^{(D)}$ formed by the columns in positions $1\leq j_1<j_2<\cdots <j_{D+1}\leq (D+1)(k+1)$ and assume that $j_2\leq (k+1), j_3\leq 2(k+1), \ldots, j_{D+1}\leq D(k+1)$.

 We construct a minor $M'$ by removing from $M$ any column which is in position $s(k+1)$ and the corresponding row $s$. As before, it is clear that
  \[M=\left( \begin{array}{c|c}
 I_{D+1-p} & \star \\
 \hline
 0 & M' \end{array}\right)\]
 where $D+1-p$ is the number of removed columns and $p$ is the size of the remaining minor $M'$.

 With a careful analysis similar to the one done in the reciprocal part of the proof, one can prove that $M'$ is a proper minor in $H'{(D)}$ and hence non zero. For this we will continue using the same notations as in the demonstration of the reverse.

Consider that the rows remaining in $M'$ are $1\leq i'_1<i'_2<\cdots i'_p=D+1$. We call this set of indices $F'$ as before. The other rows correspond to the identity columns that have been suppressed.

The corresponding column indices in $M'$ are $1\leq j'_1< j'_2 < \cdots <j_{D+1}$, and each of those columns $j'_t$ is a copy of a column $j_{f(t)}$ in $M$ and it is clear that $j_{f(t)}\in C_{b(t)}$ for some $b(t)\leq D+1$, implies $j'_t\in C'_{b(t)}$. Using the same notations as in the other part of the proof, if $j_{f(t)}=q_{f(t)}(k+1)+r_{f(t)}$ with $0\leq q_{f(t)}\leq D$ and $1\leq r_{f(t)}\leq k$, then $j_t=q_{f(t)}k+r_{f(t)}$, so they will be in the same block of column indices;  $b(t)=q_{f(t)}+1$. Note that $r(t)\leq k$ since column indices that are multiples of $k+1$ will be removed and will never turn into columns in $M'$.

\begin{itemize}

\item First we observe that $i_p=D+1$ because there were no columns of $M$ in the block with indices in $C_{D+1}$, hence the last column of $I_{D+1}$ cannot be removed (it was never there) and row $D+1$ remains in $F'$.

The corresponding column  $j'_p$ will be a  copy of some column $j_{f(p)}\leq j_{D+1}\in C_1\cup \cdots \cup C_D$, so $j'_p\in C'_1\cup \cdots \cup C'_D$, hence $j'_p\leq Dk<(D+1)k\leq i'_p k$. So the proper condition is satisfied for the last index.

\item In general when we consider row index in position $p-s$ we have $i'_{p-s}=D+1-s-r$ where $r$ is the number of identity columns after the $D+1-s$ that have been removed.

Column $j'_{p-s}$ is copy of column $j_{f(p-s)}$ and $f(p-s)\leq D+1-s-r$ ($s$ columns after it have been already considered and $r$ have been removed. From here we have $j_{f(p-s)}\leq j_{D+1-s-r}\in C_1\cup \cdots \cup C_{D-s-r}$ and this implies $j'_{p-s}\leq (D-s-r)k<i'_{p-s}k$.

\item A final observation is that $j'_1$ is always in block $C'_1$ (because block $C_1$ contained at least two columns of $M$, so even is one is removed there will always be at least one column remaining in that first block. On the other hand $i'_1\geq 1$ and we have $j'_1\leq k\leq ki'_1$.

    \end{itemize}

The first and last observations are not necessary but they help to understand the general case.

We have proven that minor $M'$ is proper and therefore cannot be singular.

\end{enumerate}

\end{proof}

\end{document}